\documentclass{llncs}

\usepackage{amssymb}
\usepackage{graphicx}
\usepackage{color}
\usepackage{textcomp}

\usepackage[utf8]{inputenc}
\usepackage{url}

\newenvironment{algorithm}{\par\vspace*{-2mm}\quad\begin{minipage}[t]{0.925\textwidth}\begin{tabbing}%
99:x\=\quad\=\quad\=\quad\=\quad\=\quad\=\kill}{\end{tabbing}\end{minipage}\par\vspace*{-3mm}}
\newcommand{\ASSERT}{\textbf{assert}\ }
\newcommand{\IF}{\textbf{if}\ }

\newcommand{\ELSE}{\textbf{else}\ }
\newcommand{\AND}{\textbf{and}\ }

\newcommand{\sequence}[1]{\left\langle#1\right\rangle}
\newcommand{\set}[1]{\left\{#1\right\}}

\newcommand{\rank}{\mbox{$\mathit{rank}$}}

\newcommand{\Construct}{\mbox{$\mathit{construct}$}}
\newcommand{\Destroy}{\mbox{$\mathit{destroy}$}}
\newcommand{\Findmin}{\mbox{$\mathit{find}$\mbox{\rm -}$\mathit{min}$}}
\newcommand{\Insert}{\mbox{$\mathit{insert}$}}
\newcommand{\Deletemin}{\mbox{$\mathit{delete}$\mbox{\rm -}$\mathit{min}$}}
\newcommand{\Decrease}{\mbox{$\mathit{decrease}$}}
\newcommand{\Delete}{\mbox{$\mathit{delete}$}}
\newcommand{\Meld}{\mbox{$\mathit{meld}$}}

\newcommand{\Access}{\mbox{$\mathit{access}$}}
\newcommand{\Grow}{\mbox{$\mathit{grow}$}}
\newcommand{\Shrink}{\mbox{$\mathit{shrink}$}}

\newcommand{\circled}[1]{\raisebox{0.14ex}[0.14ex][0.14ex]{\textbigcircle}\hspace*{-0.82em}#1}

\begin{document}

\title{Worst-Case Optimal Priority Queues via Extended Regular Counters}

\author{Amr Elmasry \and Jyrki Katajainen}

\institute{Department of Computer Science, University of Copenhagen, Denmark}

\date{}
\maketitle 
\pagestyle{plain}
\pagenumbering{arabic}

\begin{abstract}
We consider the classical problem of representing a collection of
priority queues under the operations \Findmin{}, \Insert{},
\Decrease{}, \Meld{}, \Delete{}, and \Deletemin{}.  In the
comparison-based model, if the first four operations are to be
supported in constant time, the last two operations must
take at least logarithmic time. Brodal showed that his
worst-case efficient priority queues achieve these worst-case bounds.
Unfortunately, this data structure is involved and the time bounds hide large
constants. We describe a new variant of the worst-case efficient priority 
queues that relies on extended regular counters and provides the same 
asymptotic time and space bounds as the original. Due to the conceptual 
separation of the operations on regular counters and all other operations, 
our data structure is simpler and easier to describe and understand. 
Also, the constants in the time and space bounds are smaller.
In addition, we give an implementation of our structure on a pointer machine.  
For our pointer-machine implementation, \Decrease{} and \Meld{} are asymptotically 
slower and require $O(\lg\lg{n})$ worst-case time, where $n$ denotes the number 
of elements stored in the resulting priority queue.
\end{abstract}

\section{Introduction}

A priority queue is a fundamental data structure that maintains a set of elements
and supports the operations \Findmin{}, \Insert{}, \Decrease{}, \Delete{}, 
\Deletemin{}, and \Meld{}.  
In the comparison-based model, from the $\Omega(n \lg n)$ lower bound for
sorting it follows that, if \Insert{} can be performed in $o(\lg n)$
time, \Deletemin{} must take $\Omega(\lg n)$ time. Also, if \Meld{}
can be performed in $o(n)$ time, \Deletemin{} must take $\Omega(\lg
n)$ time \cite{Bro95}.  In addition, if \Findmin{} can be performed in
constant time, \Delete{} would not be asymptotically faster than
\Deletemin{}.  Based on these observations, a priority queue is said
to provide \emph{optimal time bounds} if it can support \Findmin{},
\Insert{}, \Decrease{}, and \Meld{} in constant time; and \Delete{} and 
\Deletemin{} in $O(\lg n)$ time, where $n$ denotes the number of elements stored.

After the introduction of binary heaps \cite{Wil64}, which are not
optimal with respect to all priority-queue operations, an important
turning point was when Fredman and Tarjan introduced Fibonacci heaps
\cite{FT87}. Fibonacci heaps provide optimal time bounds for all
standard operations in the amortized sense.  Driscoll et
al.~\cite{DGST88} introduced run-relaxed heaps, which have optimal
time bounds for all operations in the worst case, except for
\Meld{}. Later, Kaplan and Tarjan \cite{KT99} (see also \cite{KST02})
introduced fat heaps, which guarantee the same worst-case bounds as
run-relaxed heaps.  On the other side, Brodal \cite{Bro95} introduced
meldable priority queues, which provide the optimal worst-case time
bounds for all operations, except for \Decrease{}. Later, by
introducing several innovative ideas, Brodal \cite{Bro96} was able to
achieve the worst-case optimal time bounds for all operations.  Though
deep and involved, Brodal's data structure is complicated and should
just be taken as a proof of existence.
Kaplan et al.~\cite{KST02} said the following about Brodal's construction: 
``This data structure is very complicated however, much more complicated 
than Fibonacci heaps and the other meldable heap data structures''.  
To appreciate the conceptual simplicity of our construction, we urge the 
reader to scan through Brodal's paper \cite{Bro96}.  On the other side, 
we admit that while trying to simplify Brodal's construction, we had to 
stick with many of his innovative ideas.
We emphasize that in this paper we are mainly interested in the
theoretical performance of the priority queues discussed.  However,
some of our ideas may be of practical value.

Most priority queues with worst-case constant \Decrease{}, including
the one to be presented by us, rely on the concept of \emph{violation
  reductions}.  A \emph{violation} is a node that may, but not necessarily, violate the heap
order by being smaller than its parent.  When the number of violations
becomes high, a violation reduction is performed in constant time to
reduce the number of violations.

A \emph{numeral system} is a notation for representing numbers 
in a consistent manner using symbols---\emph{digits}. 
In addition, operations on these numbers, as
increments and decrements of a given digit, must obey the rules
governing the numeral system.  There is a connection between numeral systems and
data-structural design \cite{CK77,Vui78}.  The idea is to relate the
number of objects of a specific type in the data structure to the
value of a digit.  
A representation of a number that is subject to increments and
decrements of arbitrary digits can be called a \emph{counter}. A
\emph{regular counter} \cite{CK77} uses the digits $\set{0, 1, 2}$ in
the representation of a number and imposes the rule that between any
two $2$'s there must be a $0$. Such a counter allows for increments 
(and also decrements, under the assumption that the digit being decreased was
non-zero) of arbitrary digits with a constant number 
of digit changes per operation. For an integer $b \geq 2$, an 
\emph{extended regular $b$-ary counter} uses the digits $\{0,\ldots,b, b+1\}$ 
with the constraints that between any two $(b+1)$'s there is a digit other 
than $b$, and between any two $0$'s there is a digit other than $1$.  
An extended regular counter \cite{CK77,KST02} allows for increments and decrements 
of arbitrary digits with a constant number of digit changes per operation.

Kaplan and Tarjan \cite{KT99} (see also \cite{KST02}) introduced fat
heaps as a simplification of Brodal's worst-case optimal priority
queues, but these are not meldable in $O(1)$ worst-case time. In fat
heaps, an extended regular ternary counter is used to maintain the
trees in each heap and an extended regular binary counter to maintain
the violation nodes. In this paper, we describe yet another
simplification of Brodal's construction.  One of the key ingredients 
in our construction is the utilization of extended regular binary
counters.  Throughout our explanation of the data structure, in
contrary to \cite{Bro96}, we distinguish between the operations of
the numeral system and other priority-queue operations.

Our motives for writing this paper were the following. 
\begin{enumerate}
\item We simplify Brodal's construction and devise a priority queue
  that provides optimal worst-case bounds for all operations (\S $2$ and \S $3$). 
  The gap between the description complexity of worst-case optimal priority queues 
  and binary heaps \cite{Wil64} is huge.  One of our motivations was to narrow this gap.
\item We describe a strikingly simple implementation of the extended regular counters (\S 4).  
  In spite of their importance for many applications, the existing descriptions \cite{CK77,KST02}
  for such implementations are sketchy and incomplete. 
\item With this paper, we complete our research program on the
  comparison complexity of priority-queue operations. All the obtained
  results are summarized in Table~\ref{table:results}.  Elsewhere, it
  has been documented that, in practice, worst-case efficient priority
  queues are often outperformed by simpler non-optimal priority queues.
  Due to the involved constant factors in the number of element comparisons, 
  this is particularly true if one aims at developing priority queues that 
  achieve optimal time bounds for all the standard operations.
\end{enumerate} 

\begin{table}[tb]
\caption{The best-known worst-case comparison complexity of different 
	priority-queue operations. The worst-case performance of \Delete{} is the same as 
	that of \Deletemin{}. 
	Using '\textbf{--}' indicates that the operation is not supported optimally.
\label{table:results}}
  
\begin{center}
\begin{tabular}{@{}|p{5.5cm}|c|c|c|c|c|@{}}
\hline
\textbf{Data structure} 
& \Findmin{}
& \Insert{}
& \Decrease{}
& \Meld{}
& \Deletemin{}\\
\hline
Multipartite priority queues \cite{EJK08b}
& $O(1)$
& $O(1)$
& \textbf{--}
& \textbf{--}
& $\lg n + O(1)$\\
\hline
Two-tier relaxed heaps \cite{EJK08a}
& $O(1)$
& $O(1)$
& $O(1)$
& \textbf{--}
& $\lg n + O(\lg\lg n)$\\
\hline
Meldable priority queues \cite{EJK10b}
& $O(1)$
& $O(1)$
& \textbf{--}
& $O(1)$
& $2\lg n + O(1)$\\
\hline
Optimal priority queues [this paper]
& $O(1)$
& $O(1)$
& $O(1)$
& $O(1)$
& $\approx 70\lg n$\\
\hline
\end{tabular}
\end{center}
\vspace{-.2in}
\end{table}

It is a long-standing open issue how to implement a heap on a
pointer machine such that
all operations are performed in optimal time bounds.
A Fibonacci heap is known to achieve the optimal bounds in the
amortized sense on a pointer machine \cite{KT08}, a fat heap in 
the worst case provided that \Meld{} is not supported \cite{KT99},
and the meldable heap in \cite{Bro95} provided that \Decrease{} is supported in $O(\lg n)$ time.  
In this paper, we offer a pointer-machine implementation for which 
\Decrease{} and \Meld{} are supported in $O(\lg \lg n)$ worst-case time.

\section{Description of the Data Structure}

Let us begin with a high-level description of the data structure.
The overall construction is similar to that used in
\cite{Bro96}.  However, the use of extended regular counters is new.
In accordance, the rank rules, and hence the structure, are different from \cite{Bro96}.
The set of violation-reduction routines are in turn new.

\begin{figure}[!tb]
\begin{center}
\input{x.pdf}
\end{center}
\begin{description}
\item[{\normalfont\circled{1}}~\,] If $t_2$ exists, $\mathit{rank}(t_1) <
  \mathit{rank}(t_2)$.

\item[{\normalfont\circled{2}}~\,] An
  extended regular binary counter is used to keep track of the children
  of each of $t_1$ and $t_2$.

\item[{\normalfont\circled{3}}~\,] For each node, including $t_1$ and
  $t_2$, its rank sequence $\sequence{d_0,d_1,\ldots,d_{\ell -1}}$
  must obey the following rules: for all $i\in\set{0,1,\ldots,\ell-1}$
  (i) $d_i \leq 3$ and (ii) if $d_i \neq 0$, then $d_{i-1} \neq 0$ or
  $d_i \geq 2$ or $d_{i+1} \neq 0$.

\item[{\normalfont\circled{4}}~\,] Each node guards a list of violations;
  $t_1$ guards two: one containing active violations and another
  containing inactive violations.

\item[{\normalfont\circled{5}}~\,] The active violations of $t_1$ are kept 
   in a violation structure consisting of a resizable array, in which the $r$th entry 
   refers to violations of rank $r$, and a doubly-linked 
   list linking the entries of the array that have more than two violations.
\end{description}
\vspace{-.2in}
\caption{Illustrating the data structure in an abstract form.\label{fig:abstract}}
\vspace{-.2in}
\end{figure}

For an illustration of the data structure, see Fig.~\ref{fig:abstract}.
Each priority queue is composed of two multi-way trees $T_1$ and
$T_2$, with roots $t_1$ and $t_2$ respectively ($T_2$ can be empty).   
The atomic components of the priority queues are nodes, each storing a single element.  
The \emph{rank} of a node $x$, denoted $\rank{}(x)$, is an integer 
that is logarithmic in the size (number of nodes) of the subtree rooted at $x$. 
The rank of a subtree is the rank of its root.
The trees $T_1$ and $T_2$ are heap ordered, except for some violations; 
but, the element at $t_1$ is the minimum among the elements of the
priority queue. If $t_2$ exists, the rank of $t_1$ is less than that of $t_2$, 
i.e.~$\rank{}(t_1) < \rank{}(t_2)$.  

With each priority queue, a \emph{violation structure} is maintained;
an idea that has been used before in \cite{Bro96,DGST88,KST02}. 
This structure is composed of a resizable array, called \emph{violation array}, in which the $r$th entry 
refers to violations of rank $r$, and a doubly-linked list, which links 
the entries of the array that have more than two violations each.  
Similar to \cite{Bro96}, each violation is guarded by a node that has a smaller element. 
Hence, when a minimum is deleted, not all the violations need to be considered 
as the new minimum candidates.  
To adjust this idea for our purpose, the violations guarded by $t_1$ 
are divided into two groups: the so-called \emph{active violations} are 
used to perform violation reductions, and the so-called \emph{inactive violations} 
are violations whose ranks were larger than the size of the violation array 
at the time when the violation occurred.  
All the active violations guarded by $t_1$, besides being kept in a doubly-linked list, 
are also kept in the violation structure; all inactive violations are kept in a doubly-linked list. 

Any node $x$ other than $t_1$ only guards a single 
list of $O(\lg n)$ violations. These are the violations that took place while 
node $x$ stored the minimum of its priority queue. Such violations must be 
tackled once the element associated with node $x$ is deleted.  
In particular, in the whole data structure there can be up to $O(n)$ violations, 
not $O(\lg n)$ violations as in run-relaxed heaps and fat heaps. 

When a violation is introduced, there are two phases for handling it accordingly to whether the size of the violation array is as big as the largest rank or not. During the first phase, the following actions are taken.
1) The new violation is added to the list of inactive violations.
2) The violation array is extended by a constant number of entries.
During the second phase, the following actions are taken.
1) The new violation is added to the list of active violations and to the violation structure.
2) A violation reduction is performed if possible.

The children of a node are stored in a doubly-linked list in
non-decreasing rank order. In addition to an element, each node stores
its rank and six pointers pointing to: the left sibling, the right
sibling (the parent if no right sibling exists), the last child (the rightmost child), 
the head of the guarded violation list, and the predecessor and the
successor in the violation list where the node may be in. 
To decide whether the right-sibling pointer of a node $x$ points to a sibling or a
parent, we locate the node $y$ pointed to by the right-sibling pointer of $x$ and check if the
last-child pointer of $y$ points back to $x$.

Next, we state the rank rules implying the structure of the multi-way trees:

\begin{description}
\item[(a)] The rank of a node is one more than the rank of its last
  child.  The rank of a node that has no children is $0$.

\item[(b)] The \emph{rank sequence} of a node specifies the
  multiplicities of the ranks of its children. If the rank sequence
  has a digit $d_r$, the node has $d_r$ children of rank $r$. The rank
  sequences of $t_1$ and $t_2$ are maintained in a way that allows
  adding and removing an arbitrary subtree of a given rank in constant
  time.  This is done by having the rank sequences of those nodes obey
  the rules of a numeral system that allows increments and decrements
  of arbitrary digits with a constant number of digit changes per
  operation.  When we add a subtree or remove a subtree from
  \emph{below} $t_1$ or $t_2$, we also do the necessary actions to
  reestablish the constraints imposed by the numeral system.

\item[(c)] Consider a node $x$ that is not $t_1$, $t_2$, or a child of
  $t_1$ or $t_2$.  If the rank of $x$ is $r$, there must exist at
  least one sibling of $x$ whose rank is $r-1$, $r$ or $r+1$.  Note
  that this is a relaxation to the rules applied to the rank sequences
  of $t_1$ and $t_2$, for which the same rule also applies
  \cite{CK77}.  In addition, the number of siblings having the same
  rank is upper bounded by at most three.
\end{description}

Among the children of a node, there are consecutive subsequences of
nodes with consecutive, and possibly equal, ranks.  We call each
maximal subsequence of such nodes a \emph{group}.  By our rank rules,
a group has at least two \emph{members}.  The difference between the
rank of a member of a group and that of another group is at least two,
otherwise both constitute the same group.

\begin{lemma}
\label{log}
The rank and the number of children of any node in our data structure is
$O(\lg{n})$, where $n$ is the size of the priority queue.
\end{lemma}

\begin{proof}
We prove by induction that the size of a subtree of rank $r$ is at
least $F_r$, where $F_r$ is the $r$th Fibonacci number. The claim
is clearly true for $r\in\set{0,1}$.  Consider a node $x$ of rank $r
\geq 2$, and assume that the claim holds for all values smaller than
$r$. The last child of $x$ has rank $r-1$. Our rank rules imply
that there is another child of rank at least $r-2$. Using the
induction hypothesis, the size of these two subtrees is at least
$F_{r-1}$ and $F_{r-2}$.  Then, the size of the subtree rooted
at $x$ is at least $F_{r-1} + F_{r-2} = F_r$.  
Hence, the maximum rank of a node is $1.44 \lg n$. 
By the rank rules, every node has at most three children of the same rank. 
It follows that the number of children per node is $O(\lg n)$.
\qed
\end{proof}

Two trees of rank $r$ can be \emph{joined} by making the tree whose
root has the larger value the last subtree of the other. The rank of
the resulting tree is $r+1$.  Alternatively, a tree rooted at a node
$x$ of rank $r+1$ can be \emph{split} by detaching its last subtree.
If the last group among the children of $x$ now has one member, the
subtree rooted at this member is also detached.  The rank of $x$
becomes one more than the rank of its current last child. In
accordance, two or three trees result from a {\it split}; among them, one
has rank $r$ and another has rank $r-1$, $r$, or $r+1$.  The {\it join} and
{\it split} operations are used to maintain the constraints imposed by the
numeral system.  Note that one element comparison is performed with
the {\it join} operation, while the {\it split} operation involves no element
comparisons.

\section{Priority-Queue Operations}

One complication, also encountered in \cite{Bro96}, 
is that not all violations can be recorded in the violation structure.
The reason is that, after a \Meld{}, the violation array may be 
too small when the old $t_1$ with the
smaller rank becomes the new $t_1$ of the melded priority queue.
Assume that we have a violation array of size $s$ associated with
$t_1$.  The priority queue may contain nodes of rank $r \geq s$.
Hence, violations of rank $r$ cannot be recorded in the
array.  We denote the violations that are recorded in the violation
array as active violations and those that are only in the
violation list as inactive violations.  Violation reductions
are performed on active violations whenever possible.  Throughout the
lifetime of $t_1$, the array is incrementally extended by the upcoming
priority-queue operations until its size reaches the largest rank.
Once the array is large enough, no new inactive violations are
created.  Since each priority-queue operation can only create a
constant number of violations, the number of inactive violations is
$O(\lg n)$.

The violation structures can be realized by letting each node have one
pointer to its violation list, and two pointers to its predecessor and
successor in the violation list where the node itself may be in. By
maintaining all active violations of the same rank consecutively in
the violation list of $t_1$, the violation array can just have a
pointer to the first active violation of any particular rank.

In connection with every \Decrease{} or \Meld{}, if $T_2$ exists, a constant number of subtrees 
rooted at the children of $t_2$ are removed from below $t_2$ and added below $t_1$.  
Once $\rank{}(t_1)\geq \rank{}(t_2)$, the whole tree $T_2$ is added below $t_1$.
To be able to move all subtrees from below $t_2$ and
finish the job on time, we should always pick a subtree
from below $t_2$ whose rank equals the current rank of $t_1$.  

The priority queue operations aim at maintaining the following invariants:

\begin{enumerate}
\vspace{-.1in}
\item
The minimum is associated with $t_1$.  
\item
The second-smallest element is either stored at $t_2$, at one of the children of
$t_1$, or at one of the violation nodes associated with $t_1$.
\item
The number of entries in the violation list of a node is $O(\lg n)$,
assuming that the priority queue that contains this node has $n$ elements.
\vspace{-.05in}
\end{enumerate}

We are now ready to describe how the priority-queue operations are realized.
\begin{description}
\vspace{-.05in}
\item[\Findmin{}$(Q)$:] Following the first invariant, the minimum is at $t_1$.

\item[\Insert{}$(Q,x)$:] A new node $x$ is given with the value
  $e$. If $e$ is smaller than the value of $t_1$, the
  roles of $x$ and $t_1$ are exchanged by swapping the two nodes. The
  node $x$ is then added below $t_1$.

\item[\Meld{}$(Q, Q')$:] This operation involves at most four trees $T_1$, $T_2$, $T_1'$, and $T_2'$,
  two for each priority queue; their roots are named correspondingly using lower-case letters. 
  Assume without loss of generality that $value(t_1) \leq value(t'_1)$.
  The tree $T_1$ becomes the first tree of the melded priority queue.  
  The violation array of $t'_1$ is dismissed.  
  If $T_1$ has the maximum rank, the other trees are added below $t_1$ 
  resulting in no second tree for the melded priority queue.  
  Otherwise, the tree with the maximum rank among $T'_1, T_2$, and $T'_2$
  becomes the second tree of the melded priority queue. The remaining trees are
  added below the root of this tree, and
  the roots of the added trees are made violating. To keep the number
  of active violations within the threshold, two violation
  reductions are performed if possible. Finally, the regular counters
  that are no longer corresponding to roots are dismissed.

\item[\Decrease{}$(Q,x,e)$:] The element at node $x$ is replaced by
  element $e$. If $e$ is smaller than the element at $t_1$, the roles
  of $x$ and $t_1$ are exchanged by swapping the two nodes (but not
  their violation lists).  If $x$ is either $t_1$, $t_2$, or a
  child of $t_1$, stop.  Otherwise, $x$ is denoted violating and added
  to the violation structure of $t_1$; if $x$ was already violating,
  it is removed from the violation structure where it was in.  To keep
  the number of active violations within the threshold, a violation
  reduction is performed if possible.

\item[\Deletemin{}$(Q)$:] By the first invariant, the minimum is at $t_1$.
	The node $t_2$ and all the subtrees rooted at its 
	children are added below $t_1$. This is accompanied with extending the violation 
	array of $T_1$, and dismissing the regular counter of $T_2$.  
  By the second invariant, the new minimum is now stored at one of the children
  or violation nodes of $t_1$.  By Lemma \ref{log} and the third
  invariant, the number of minimum candidates is $O(\lg{n})$.  Let $x$ be
  the node with the new minimum. If $x$ is among the violation nodes
  of $t_1$, a tree that has the same rank as $x$ is removed from below $t_1$,
  its root is made violating, and is attached in place of the subtree rooted at $x$.
  If $x$ is among the children of $t_1$, the tree rooted at $x$ is removed from below $t_1$. 
  The inactive violations of $t_1$ are recorded in the violation array. The violations of $x$
  are also recorded in the array.  
  The violation list of $t_1$ is appended to that of $x$. 
  The node $t_1$ is then deleted and replaced by the node $x$.
  The old subtrees of $x$ are added, one by one, below the new root.  
  To keep the number of violations within the threshold, violation reductions are performed as many times as possible.
  By the third invariant, at most $O(\lg{n})$ violation reductions are to be performed.

\item[\Delete{}$(Q,x)$:] The node $x$ is swapped with $t_1$, which is then made violating. 
	To remove the current root $x$, the same actions are performed as in \Deletemin{}.
\vspace{-.2in}	
\end{description}
  
In our description, we assume that it is possible to dismiss an array
in constant time.  We also assume that the doubly-linked list indicating
the ranks where a reduction is possible is realized inside the
violation array, and that a regular counter is compactly
represented within an array.  Hence, the only garbage created by
freeing a violation structure or a regular counter is an array of
pointers.  If it is not possible to dismiss an array in constant
time, we rely on incremental garbage collection.  In such case, to
dismiss a violation structure or a regular counter, we add it
to the garbage pile, and release a constant amount of garbage in
connection with every priority-queue operation.  It is not hard to
prove by induction that the sum of the sizes of the violation
structures, the regular counters, and the garbage pile
remains linear in the size of the priority queue.

\subsection*{Violation Reductions}

Each time when new violations are introduced, we perform equally many
violation reductions whenever possible. A violation reduction is
possible if there exists a rank recording at least three active
violations. This will fix the maximum number of active violations at
$O(\lg n)$. Our violation reductions diminish the number of violations
by either getting rid of one violation, getting rid of two and
introducing one new violation, or getting rid of three and introducing
two new violations.  We use the powerful tool that the rank sequence
of $t_1$ obey the rules of a numeral system, which allows adding a new
subtree or removing an existing one from below $t_1$ in worst-case
constant time.  When a subtree with a violating root is added below
$t_1$, its root is no longer violating.

One consequence of allowing $O(n)$ violations is that we cannot use
the violation reductions exactly in the form described for run-relaxed
heaps or fat heaps.  When applying the cleaning transformation
to a node of rank $r$ (see \cite{DGST88} for the details), we cannot
any more be sure that its sibling of rank $r+1$ is not
violating, simply because there can be violations that are guarded by other nodes.  
We then have to avoid executing the cleaning transformation by the violation reductions.
 
Let $x_1$, $x_2$, and $x_3$ be three violations of the same rank $r$. 
We distinguish several cases to be considered when applying our reductions:

\begin{description}
\vspace{-.1in}
\item[Case 1.] If, for some $i \in\{1,2,3\}$, $x_i$ is neither the last nor the 
	second-last child, detach the subtree rooted at $x_i$ from its parent and add it below $t_1$. 
	The node $x_i$ will not be anymore violating. The detachment of $x_i$ may leave one or two groups with one member (but not the last group).
	If this happens, the subtree rooted at each of these singleton members is then detached and added below $t_1$.
	(We can still detach the subtree of $x_i$ even when $x_i$ is one of the last two children of its parent,
	conditioned that such detachment leaves this last group with at least two members and retains the rank of the parent.) 
\end{description}

For the remaining cases, after checking Case $1$, we assume that each of $x_1,x_2$, and $x_3$ 
is either the last or the second-last child of its parent.
Let $s_1$, $s_2$, and $s_3$ be the other member of the last two
members of the groups of $x_1$, $x_2$, and $x_3$, respectively. Let
$p_1$, $p_2$, and $p_3$ be the parents of $x_1$, $x_2$, and $x_3$,
respectively.  Assume without loss of generality that $\rank{}(s_1)
\geq \rank{}(s_2) \geq \rank{}(s_3)$.

\begin{description}
\vspace{-.1in}
\item[Case 2.] 
$\rank{}(s_1) = \rank{}(s_2) = r+1$, or $\rank{}(s_1) = \rank{}(s_2) = r$, 
or $\rank{}(s_1) = r$ and $\rank{}(s_2) = r-1$: 

\begin{description}
\item[(a)] $value(p_1) \leq value(p_2)$: Detach the subtrees rooted at $x_1$ and $x_2$, and add
  them below $t_1$; this reduces the number of violations by two. 
  Detach the subtree rooted at $s_2$ and attach it below $p_1$ (retain rank order); this does not introduce any new violations.
  Detach the subtree rooted at the last child of $p_2$ if it is a singleton member, detach the remainder of the subtree rooted at $p_2$, 
  change the rank of $p_2$ to one more than that of its current last child, 
  and add the resulting subtrees below $t_1$. Remove a subtree with the old rank
  of $p_2$ from below $t_1$, make its root violating, and attach it in the old
  place of the subtree rooted at $p_2$.
\item[(b)] $value(p_2) < value(p_1)$: Change the roles of $x_1,s_1,p_1$
  and $x_2,s_2,p_2$, and apply the same actions as in Case 2(a). 
\end{description}

\item[Case 3.]
$\rank{}(s_1) = r+1$ and $\rank{}(s_2) = r$: 

\begin{description}
\item[(a)] $value(p_1) \leq value(p_2)$: Apply the same actions as in Case 2(a).
\item[(b)] $value(p_2) < value(p_1)$: Detach the subtrees rooted at $x_1$ and $x_2$, and add
  them below $t_1$; this reduces the number of violations by two. 
  Detach the subtree rooted at $s_2$ if it becomes a singleton member of its group, 
  detach the remainder of the subtree rooted at $p_2$, 
  change the rank of $p_2$ to one more than that of its current last child, 
  and add the resulting subtrees below $t_1$.
  Detach the subtree rooted at $s_1$, and attach it in the old place of the subtree
  rooted at $p_2$; this does not introduce any new violations. 
  Detach the subtree rooted at the current last child of $p_1$ if such child becomes a singleton member of its group, 
  detach the remainder of the subtree rooted at $p_1$, 
  change the rank of $p_1$ to one more than that of its current last child, 
  and add the resulting subtrees below $t_1$.
  Remove a subtree of rank $r+2$ from below $t_1$, make its root violating,
  and attach it in the old place of the subtree rooted at $p_1$.
\end{description}

\item[Case 4.] $\rank{}(s_1) = \rank{}(s_2) = \rank{}(s_3) = r-1$:

Assume without loss of generality that $value(p_1) \leq
\min\{value(p_2),value(p_3)\}$.  Detach the subtrees of $x_1, x_2$, and $x_3$, and add them
below $t_1$; this reduces the number of violations by three.
Detach the subtrees of $s_2$ and $s_3$, and join them to form a
subtree of rank $r$. Attach the resulting subtree in place of
$x_1$; this does not introduce any new violations. 
Detach the subtree rooted at the current last child of each of $p_2$ and $p_3$
if such child becomes a singleton member of its group,
detach the remainder of the subtrees rooted at $p_2$ and $p_3$, change the rank of each
of $p_2$ and $p_3$ to one more than that of its current last child,
and add the resulting subtrees below $t_1$. Remove two subtrees of rank $r+1$
from below $t_1$, make the roots of each of them violating, and attach them in the old
places of the subtrees rooted at $p_2$ and $p_3$.

\item[Case 5.] $\rank{}(s_1) = r+1$ and $\rank{}(s_2) = \rank{}(s_3) = r-1$:
\begin{description}
\item[(a)] $value(p_1) \leq \min\{value(p_2),value(p_3)\}$: Apply same actions as Case 4.
\item[(b)] $value(p_2) < value(p_1)$: Apply the same actions as in Case 3(b).
\item[(c)] $value(p_3) < value(p_1)$: Change the roles of
  $x_2,s_2,p_2$ to $x_3,s_3,p_3$, and apply the same actions as in Case 3(b).
\end{description}
\end{description}

The following properties are the keys for the success of our violation-reduction routines.
1) Since there is no tree $T_2$ when a violation reduction takes place, 
the rank of $t_1$ will be the maximum rank among all other nodes. 
In accordance, we can remove a subtree of any specified rank from below $t_1$.
2) Since $t_1$ has the minimum value of the priority queue, its children 
are not violating. In accordance, we can add a subtree below $t_1$ and ensure that
its root is not violating.

\section{Extended Regular Binary Counters}

An extended regular binary counter represents a non-negative integer
$n$ as a string $\sequence{d_0,d_1,\ldots,d_{\ell-1}}$ of digits,
least-sig\-nifi\-cant digit first, such that $d_i \in \{0,1,2,3\}$,
$d_{\ell-1} \neq 0$, and $n =\sum _i d_i \cdot 2^i$.  The main
constraint is that every $3$ is preceded by a $0$ or $1$ possibly
having any number of $2$'s in between, and that every $0$ is preceded
by a $2$ or $3$ possibly having any number of $1$'s in between.  This
constraint is stricter than the standard one, which allows the first
$3$ and the first $0$ to come after any (or even no) digits.  An
extended regular counter \cite{CK77,KST02} supports the increments and
decrements of arbitrary digits with a constant number of digit changes
per operation.

Brodal \cite{Bro96} showed how, what he calls a \emph{guide},
can realize a regular binary counter (the digit set has three symbols
and the counter supports {\it increments} in constant time).  To support
{\it decrements} in constant time as well, he suggested to couple two such
guides back to back.  We show how to implement an extended regular
binary counter more efficiently using much simpler ideas. The
construction described here was sketched in \cite{KST02}; our
description is more detailed.

We partition any sequence into \emph{blocks} of consecutive digits,
and digits that are in no blocks.  We have two categories of blocks:
blocks that end with a $3$ are of the forms $12^*3, 02^*3$, and blocks
that end with a $0$ are of the forms $21^*0, 31^*0$ (we assume that
least-significant digits come first, and $^{*}$ means zero or more
repetitions).  We call the last digit of a block the
\emph{distinguishing digit}, and the other digits of the block the
\emph{members}.  Note that the distinguishing digit of a block may be
the first member of a block from the other category.

To efficiently implement \emph{increment} and \emph{decrement} operations, a
\emph{fix} is performed at most twice per operation.  A fix does not
change the value of a number.  When a digit that is a member of a
block is increased or decreased by one, we may need to perform a fix
on the distinguishing digit of its block.  We associate a forward
pointer $f_i$ with every digit $d_i$, and maintain the invariant that
all the members of the same block point to the distinguishing digit of
that block. The forward pointers of the digits that are not members of 
a block point to an arbitrary digit.  Starting from any member, 
we can access the distinguishing digit of its block, 
and hence perform the required fix, in constant time.
 
As a result of an \emph{increment} or a \emph{decrement}, the 
following properties make such a construction possible.
\begin{itemize}
\vspace{-.1in}
\item A block may only extend from the beginning and by only one digit.
In accordance, the forward pointer of this new member inherits the same value as the forward pointer of the following digit.
\item A newly created block will have only two digits.
In accordance, the forward pointer of the first digit is made to point to the other digit.
\item A fix that is performed unnecessarily is not harmful (keeps the representation regular). 
In accordance, if a block is destroyed when fixing its distinguishing digit, no changes are done with the forward pointers.
\vspace{-.1in}
\end{itemize}
A string of length zero represents the number $0$. In our pseudo-code,
the change in the length of the representation is implicit; 
the key observation is that the length can only increase by at most one digit 
with an \emph{increment} and decrease by at most two digits with a \emph{decrement}. 

For $d_j=3$, a \emph{fix-carry} for $d_j$ is performed as follows: 
\begin{center}
\begin{algorithm}
\rule[6pt]{\textwidth}{0.45pt}\\[-6pt]
\textbf{Algorithm} \textit{fix-carry}($\sequence{d_0,d_1,\ldots,d_{\ell-1}}$, $j$)\\
\rule[6pt]{\textwidth}{0.45pt}\\[-6pt]
1:\>\ASSERT $0 \leq j \leq \ell-1$ \AND $d_j=3$\\
2:\>$d_j \leftarrow 1$\\
3:\>increase $d_{j+1}$ by $1$\\
4:\>\IF $d_{j+1} = 3$\\ 
5:\>\> $f_j \leftarrow j+1$ ~~// a new block of two digits\\
6:\>\ELSE\\
7:\>\>$f_j \leftarrow f_{j+1}$ ~~// extending a \emph{possible} block from the beginning\\ 
\rule[6pt]{\textwidth}{0.45pt}
\end{algorithm}
\end{center}
As a result of a {\it fix-carry} the value of a number does not change.
Accompanying a {\it fix-carry}, two digits are changed. In the corresponding
data structure, this results in performing a \emph{join}, which
involves one element comparison.

The following pseudo-code summarizes the actions needed to increase
the $i$th digit of a number by one.

\begin{center}
\begin{algorithm}
\rule[6pt]{\textwidth}{0.45pt}\\[-6pt]
\textbf{Algorithm} \textit{increment}$(\sequence{d_0,d_1,\ldots,d_{\ell-1}}, i)$\\
\rule[6pt]{\textwidth}{0.45pt}\\[-6pt]
1:\>\ASSERT $0 \leq i \leq \ell$\\
2:\>\IF $d_i = 3$\\
3:\>\>\textit{fix-carry}($\sequence{d_0,d_1,\ldots,d_{\ell-1}}$, $i$)
~~// either this fix is executed\\
6:\>$j \leftarrow f_i$\\
7:\>\IF $j \leq \ell-1$ \AND $d_j = 3$\\
8:\>\>\textit{fix-carry}($\sequence{d_0,d_1,\ldots,d_{\ell-1}}$, $j$)\\
8:\>increase $d_i$ by $1$\\
10:\>\IF $d_i = 3$\\
11:\>\> \textit{fix-carry}($\sequence{d_0,d_1,\ldots,d_{\ell-1}}$,
$i$) ~~// or this fix\\
\rule[6pt]{\textwidth}{0.45pt}
\end{algorithm}
\end{center}

\noindent
Using case analysis, it is not hard to verify that this operation maintains the
regularity of the representation.

For $d_j=0$, a \emph{fix-borrow} for $d_j$ is performed as follows: 

\begin{center}
\begin{algorithm}
\rule[6pt]{\textwidth}{0.45pt}\\[-6pt]
\textbf{Algorithm} \textit{fix-borrow}($\sequence{d_0,d_1,\ldots,d_{\ell-1}}$, $j$)\\
\rule[6pt]{\textwidth}{0.45pt}\\[-6pt]
1:\>\ASSERT $0 \leq j < \ell-1$ \AND $d_j=0$\\
2:\>$d_j \leftarrow 2$\\
3:\>decrease $d_{j+1}$ by $1$\\
4:\>\IF $d_{j+1} = 0$\\ 
5:\>\> $f_j \leftarrow j+1$ ~~// a new block of two digits\\
6:\>\ELSE\\
7:\>\>$f_j \leftarrow f_{j+1}$ ~~// extending a \emph{possible} block from the beginning\\ 
\rule[6pt]{\textwidth}{0.45pt}
\end{algorithm}
\end{center}

\noindent
As a result of a {\it fix-borrow} the value of a number does not change.
Accompanying a {\it fix-borrow}, two digits are changed. In the
corresponding data structure, this results in performing a
\emph{split}, which involves no element comparisons.

The following pseudo-code summarizes the actions needed to decrease
the $i$th digit of a number by one.

\begin{center}
\begin{algorithm}
\rule[6pt]{\textwidth}{0.45pt}\\[-6pt]
\textbf{Algorithm} \textit{decrement}$(\sequence{d_0,d_1,\ldots,d_{\ell-1}}, i)$\\
\rule[6pt]{\textwidth}{0.45pt}\\[-6pt]
1:\>\ASSERT $0 \leq i \leq \ell-1$\\
2:\>\IF $d_i = 0$\\
3:\>\>\textit{fix-borrow}($\sequence{d_0,d_1,\ldots,d_{\ell-1}}$, $i$)
~~// either this fix is executed\\
6:\>$j \leftarrow f_i$\\
7:\>\IF $j < \ell-1$ \AND $d_j = 0$\\
8:\>\>\textit{fix-borrow}($\sequence{d_0,d_1,\ldots,d_{\ell-1}}$, $j$)\\
8:\>decrease $d_i$ by $1$\\
10:\>\IF $i < \ell-1$ \AND $d_i = 0$\\
11:\>\> \textit{fix-borrow}($\sequence{d_0,d_1,\ldots,d_{\ell-1}}$,
$i$) ~~// or this fix\\
\rule[6pt]{\textwidth}{0.45pt}
\end{algorithm}
\end{center}

\noindent
Using case analysis, it is not hard to verify that this operation maintains the
regularity of the representation.

In our application, the outcome of splitting a tree of rank $r+1$ 
is not always two trees of rank $r$ as assumed above. The
outcome can also be one tree of rank $r$ and another of rank $r+1$;
two trees of rank $r$ and a third tree of a smaller rank; or one
tree of rank $r$, one of rank $r-1$, and a third tree of a smaller
rank. We can handle the third tree, if any, by adding it 
below $t_1$ (executing an \emph{increment} just after the \emph{decrement}).  
The remaining two cases, where the {\it split} results in one
tree of rank $r$ and another of rank either $r+1$ or $r-1$, 
are pretty similar to the case where we have two trees of rank $r$. 
In the first case, we have detached a tree of rank $r+1$ and added a tree of rank $r$ 
and another of rank $r+1$; this case maintains the representation regular. 
In the second case, we have detached a tree of rank $r+1$ and added a tree of rank $r$ 
and another of rank $r-1$; after that, there may be three or four trees of rank $r-1$, 
and one {\it join} (executing a {\it fix-carry}) at rank $r-1$ may be necessary to make the representation regular.  
In the worst case, a \emph{fix-borrow} may require three element comparisons: 
two because of the extra addition (\emph{increment}) and one because of the extra \emph{join}.
A \emph{decrement}, which involves two fixes, may then require up to 
six element comparisons.

\section{Analysis of \Deletemin{}}

Let us analyse the number of element comparisons performed per
\Deletemin{}.  Recall that the maximum rank is bounded by $1.44 \lg
n$.  The sum of the digits for an integer obeying the extended regular
binary system is two times the maximum rank, i.e.~at most $2.88 \lg
n$.  It follows that the number of children of any node is at most
$2.88 \lg n$.  For an extended regular counter, an {\it increment} requires at
most two element comparisons and a {\it decrement} at most six element
comparisons. Recall also that the number of active violations cannot
be larger than $2.88 \lg n$, since if there were more than two
violations per rank a violation reduction would be possible. The
number of inactive violations is less than $\varepsilon \lg n$, where
$\varepsilon > 0$; its actual value depends on the speed we extend the
violation array and the number of new violations per operation. A
realistic value would be $\varepsilon = 1/10$.

First, when $t_2$ and its children are moved below $t_1$, at most
$2.88 \lg n + O(1)$ element comparisons are performed.  Second, when
finding the new minimum, during the scan over the children of $t_1$ at
most $2.88 \lg n$ element comparisons are made, and during the scan
over the violations of $t_1$ $(2.88 + \varepsilon)\lg n$ element
comparisons are made. Third, when all the inactive violations are made
active, $2\varepsilon \lg n$ element comparisons are performed.
Fourth, when the old children of $x$ are merged with those of the old
root, at most $2.88 \lg n$ element comparisons are performed.  Fifth,
when performing the necessary violation reductions, Case 4 turns out
to be the most expensive: (i) The minimum value of three elements is
to be determined (this requires two element comparisons). (ii) A {\it join}
requires one element comparison. (iii) Up to seven {\it increments} and up
to two {\it decrements} at the extended regular counter of $t_1$ may be necessary. In total, a
single violation reduction may require $29$ element comparisons, and
$(2.88 + \varepsilon) \lg n$ such reductions may take place (the
$\varepsilon$ accounts for the inactive violations of $t_1$).  To sum
up, the total number of element comparisons performed is at most $(4
\times 2.88 + \varepsilon + 29 (2.88 + \varepsilon)) \lg n + O(1)$,
which is at most $(95.04 + 30 \varepsilon) \lg n + O(1)$ element comparisons.

An extra optimization is still possible. During the violation reduction, instead 
of adding the subtrees of $x_2$ and $x_3$ below $t_1$ (an operation that requires two
element comparisons per subtree), we can join the two subtrees (an operation that 
requires one element comparison) using the fact that they have equal ranks. 
In addition, we can attach the resulting subtree in the place of the subtree that was rooted at $p_2$. 
Hence, we only need to remove one tree from below $t_1$ instead of two 
(an operation that requires six element comparisons per subtree).
In total, we perform at most one minimum-finding among three elements, 
two joins, five {\it increments} and one {\it decrement} per violation reduction. 
In accordance, a single violation reduction may require at most $20$ element comparisons. 
This reduces the number of element comparisons per \Deletemin{} to at most
$(4 \times 2.88 + \varepsilon + 20 (2.88 + \varepsilon)) \lg n + O(1) = (69.12 + 21 \varepsilon) \lg n + O(1)$ element comparisons.

\section{A Pointer-Machine Implementation}

In a nutshell, a pointer machine is a model of computation that only
allows pointer manipulations; arrays, bit tricks, or arithmetic
operations are not allowed.  In the literature, the definitions (and
names) of this type of machines differ, and it seems that all of these
machines do not have the same computational power. We want to
emphasize that the model considered here is a restricted one that may
be called a \emph{pure pointer machine} (for a longer discussion about
different versions of pointer machines and their computational power,
see \cite{Ben95}).

The memory of a pointer machine is a directed graph of cells, each
storing a constant number of pointers to other cells or to the cell
itself. All cells are accessed via special centre cells seen as
incarnations of registers of ordinary computers. The primitive
operations allowed include the creation of cells, destruction of
cells, assignment of pointers, and equality comparison of two
pointers. It is assumed that the elements manipulated can be copied
and compared in constant time, and that all primitive operations
require constant time.

In connection with the initial construction, we allocate two cells to represent the bits 0 and 1,
and three cells to represent the colours white, black, and red.  

Since we use ranks and array indices, we should be able to represent integers in
the range between 0 and the maximum rank of any priority queue.  
For this purpose, we use the so-called \emph{rank entities}; 
we keep these entities in a doubly-linked list, each entity representing one integer. 
The rank entities are shared by all the priority queues.
Every node can access the entity corresponding to its rank in
constant time by storing a pointer to this rank entity. 
With this mechanism, increments and decrements of ranks, rank copying, and equality
comparisons between two ranks can be carried out in constant time.
The maximum rank may only be increased or decreased by at most one with any operation. 
When the maximum rank is increased, a new entity is created and appended to the list. 
To decide when to release the last entity, the rank
entities are reference counted. This is done by associating with each
entity a list of cells corresponding to the nodes having this rank. 
When the last entity has no associated cells, the entity is released.
We also build a complete binary tree above the rank entities; we call
this tree the \emph{global index}.  Starting from a rank entity, the
bits in the binary representation of its rank are
obtained by accessing the global index upwards.

The main idea in our construction is to simulate a resizable array on
a pointer machine.  Given the binary representation of an array index,
which in our case is a rank, the corresponding location can be
accessed without relying on the power of a random-access machine, by
having a complete binary tree built above the array entries;
we call this tree the \emph{local index}. 
Starting from the root of the tree, we scan the bits in the representation of
the rank, and move either to the left or to the right until we reach
the sought entity that corresponds to an array entry; 
we call such entities the \emph{location entities}.   
The last location entity corresponds to the maximum node rank;
that is the rank of $t_2$ (or $t_1$ if $T_2$ is empty).

To facilitate the dynamization of indexes, we maintain four versions of
each; one is under construction, one is under destruction, and two are
complete (one of them is in use).  The sizes of the four versions are
consecutive powers of two, and there is always a version that is
larger and another that is smaller than the version in use.  At the
right moment, we make a specific version the current index, start
destroying an old version, and start constructing a new version.  To
accomplish such a change, each entity must
have three pointers, and we maintain a switch that indicates which
pointer leads to the index in use.  Index constructions and
destructions are done incrementally by distributing the work on the
operations using the index, a constant amount of work per operation.

To decide the right moment for index switching, we colour the leaves of 
every index. Consider an index with $2^k$ leaves. The rightmost leaf
is coloured red, the leaf that is at position $2^{k-2}$ from the left is coloured black, 
and all other leaves are coloured white.  
When the current index becomes full, i.e.~the red cell is met
by the last entity, we switch to the next larger index. When a
black cell is met by the last entity, we switch to the next smaller index. 

Let us now consider how these tools are used to implement our data
structure on a pointer machine.  One should be aware that the
following two operations require $O(\lg\lg n)$ worst-case time on a pointer machine: 
\begin{itemize}
\item When we access an array to record a violation or adjust the
  numeral system (in \Decrease{} or in \Meld{}), we use the global
  index to get the bits in the binary representation of the rank of
  the accessed entity and then a local index to find the corresponding
  location. We use the height of the local index to determine how many
  bits to extract from the global index. 
\item When we compare two ranks (in \Meld{}), we extract the bits from
  the global index and perform a normal comparison of two bit sequences. 
\end{itemize}
When two priority queues are melded, we keep the local index of the 
violation structure of one priority queue, and move the other to the garbage pile. 
The unnecessary regular counters are also moved to the garbage pile. As pointed out earlier,
incremental garbage collection ensures that the amount of space used
by each priority queue is still linear in its size.  Because of the rank
comparison, the update of the violation structure of $t_1$, and the
update of the regular counter of $t_2$, \Meld{} takes $O(\lg\lg n)$ worst-case
time. Also \Decrease{} involves several array accesses, so its
worst-case cost becomes $O(\lg\lg n)$.

Lastly, \Deletemin{} must be implemented carefully to avoid a slowdown
in its running time. This is done by bypassing the local indexes and
creating a one-to-one mapping between the rank entities and the
location entities maintained for the violation structure of $t_1$ and
the regular counters of $t_1$ and $t_2$.  First, the rank and location
entities are simultaneously scanned, one by one, and a pointer to the
corresponding location entity is stored at each rank entity.
Hereafter, any node keeping a pointer to a rank entity can access the
corresponding location entity with a constant delay.  This way,
deletions can be carried out asymptotically as fast as on a
random-access machine.

\section{Summarizing the Results}

We have proved the following theorem.

\begin{theorem}
There exists a priority queue that works on a random-access machine
and supports the operations \Findmin{},
\Insert{}, \Decrease{}, and \Meld{} in constant worst-case time; and
\Delete{} and \Deletemin{} in $O(\lg n)$ worst-case time, where $n$ is
the number of elements in the priority queue.  The amount of space used by the
priority queue is linear.
\end{theorem}

For our structure, the number of element comparisons performed by
\Delete{} is at most $\beta \cdot \lg n$, where $\beta$ is around
seventy. For the data structure in \cite{Bro96}, some
implementation details are not specified, making such analysis depend on
the additional assumptions made.  According to our rough estimates,
the constant factor $\beta$ for the priority queue in \cite{Bro96} is much more than
one hundred. We note that for both data structures further
optimizations may still be possible.

We summarize our storage requirements as follows. Every node stores an
element, a rank, two sibling pointers, a last-child pointer, a pointer
to its violation list, and two pointers for the violation list it may
be in. The violation structure and the regular counters require a logarithmic
amount of space.  In \cite{Bro96}, in order to keep external
references valid, elements should be stored indirectly at nodes, as
explained in \cite[Chapter 6]{CLRS09}.  The good news is that we save
four pointers per node in comparison with \cite{Bro96}. The bad news
is that, for example, binomial queues \cite{Vui78} can be implemented
with only two pointers per node.

We have also proved the following theorem.

\begin{theorem}
There exists a priority queue that works on a pure pointer machine and
supports the operations \Findmin{} and \Insert{} in constant
worst-case time, \Decrease{} and \Meld{} in $O(\lg\lg n)$ worst-case
time, and \Delete{} and \Deletemin{} in $O(\lg n)$ worst-case time,
where $n$ is the number of elements stored in the resulting priority
queue.  The amount of space used by the priority queue is linear.
\end{theorem}

\section{Conclusions}

We showed that a simpler data structure achieving the same asymptotic
bounds as Brodal's data structure \cite{Bro96} exists.  We were
careful not to introduce any artificial complications when presenting
our data structure.  Our construction reduces the constant factor
hidden behind the big-Oh notation for the worst-case number of element
comparisons performed by \Deletemin{}.  Theoretically, it would be
interesting to know how low such factor can get.

We summarize the main differences between our construction and that in
\cite{Bro96}.  On the positive side of our treatment:
\begin{itemize}

\item We use a standard numeral system with fewer digits (four instead
  of six).  Besides improving the constant factors, this allows for
  the logical distinction between the operations of
  the numeral system and other operations.

\item We use normal joins instead of three-way joins, each involving
  one element comparison instead of two.

\item We do not use parent pointers, except for the last children.
  This saves one pointer per node, and allows swapping of nodes in
  constant time. Since node swapping is possible, elements can be
  stored directly inside nodes. This saves two more pointers per node.

\item We gathered the violations associated with every node in one
  violation list, instead of two. This saves one more pointer per node.

\item We bound the number of violations by restricting their total
  number to be within a threshold, whereas the treatment
  in \cite{Bro96} imposes an involved numeral system that constrains
  the number of violations per rank.

\end{itemize}
\noindent
On the negative side of our treatment:
\begin{itemize}
\item We check more cases within our violation reductions.

\item We have to deal with larger ranks; the maximum rank may go up to
  $1.44 \lg n$, instead of $\lg n$.
\end{itemize}

It is a long-standing open issue how to implement a priority queue on
a pointer machine such that all operations are performed in optimal
time bounds.  Fibonacci heaps are known to achieve the optimal bounds
in the amortized sense on a pointer machine \cite{KT08}. Meldable
priority queues described in \cite{Bro95} can be implemented on a pointer
machine, but \Decrease{} requires $O(\log n)$ time.  Fat heaps
\cite{KST02} can be implemented on a pointer machine, but \Meld{} is
not supported.  In \cite{EJK10c}, a pointer-machine implementation of
run-relaxed weak heaps is given, but \Meld{} requires $O(\lg n)$ time
and \Decrease{} $O(\lg\lg n)$ time.  We introduced a pointer-machine
implementation that supports \Meld{} and \Decrease{} in $O(\lg\lg n)$
time so that no arrays, bit tricks, or arithmetic operations are used.
We consider the possibility of constructing worst-case optimal
priority queues that work on a pointer machine as a major open question.

\section*{Acknowledgements}

We thank Robert Tarjan for motivating us to simplify Brodal's
construction, Gerth Brodal for taking the time to explain some details
of his construction, and Claus Jensen for reviewing this manuscript
and teaming up with us for six years in our research on the comparison
complexity of priority-queue operations.

\newpage
\appendix
\section{Completing the Story}

This appendix was written for our non-expert readers, to give them
insight into the technical details---of which there are still
many---skipped in the main body of the text.  Our hope is that after
reading this appendix the reader would get the details 
used for obtaining a worst-case optimal priority queue.

\subsection{Resizable Arrays}

One complication encountered in our construction is the dynamization
of the arrays maintained. At some point, an array may be too small to
record the information it is supposed to hold. Sometimes it may also
be necessary to contract an array so that its size is proportional to
the amount of recorded information.

A standard solution for this problem is to rely on doubling, halving,
and incremental copying. Each array has a \emph{size}, i.e.~the number
of objects stored, and a \emph{capacity}, i.e.~the amount of space
allocated for it, measured in the number of objects too. The
representation of an array is composed of up to two memory segments, and
integers indicating their size and capacity. Let us call
the two segments $X$ and $Y$; $X$ is the main structure, and $Y$ is the
copy under construction, which may be empty. The currently stored
objects are found in at most two consecutive subsegments, one at the
beginning of $X$ and another in the middle of $Y$. Initially,
$\mathit{capacity}(X) = 1$ and $\mathit{size}(X) = 0$. With this organization,
the array operations are performed as follows:

\begin{description}
\item[\Access{}($A$, $i$).] To access the object with index $i$ in
  array $A$, return $X[i]$ if $i < \mathit{size}(X)$;
  otherwise, return $Y[i]$. (No out-of-bounds index checking is done.)

\item[\Grow{}($A$).] 
  If $Y$ does not exist and $\mathit{size}(X) < \mathit{capacity}(X)$, 
  increase $\mathit{size}(X)$ by one and stop. 
  If $Y$ does not exist and $\mathit{size}(X) = \mathit{capacity}(X)$, allocate space 
  for $Y$ where $\mathit{capacity}(Y) = 2 \cdot \mathit{capacity}(X)$
  and set $\mathit{size}(Y) = \mathit{size}(X)$.   
  Copy one object from the end of $X$ to the corresponding position in $Y$, 
  decrease $\mathit{size}(X)$ by one and increase $\mathit{size}(Y)$ by one. 
  If $\mathit{size}(X)$ is zero, release the array $X$ and rename $Y$ as $X$.

\item[\Shrink{}($A$).] 
  If $Y$ does not exist and $\mathit{size}(X) > 1/4\cdot \mathit{capacity}(X)$, 
  decrease $\mathit{size}(X)$ by one and stop. 
  If $Y$ does not exist and $\mathit{size}(X) = 1/4\cdot \mathit{capacity}(X)$, 
  allocate space for $Y$ where $\mathit{capacity}(Y) = 1/2 \cdot\allowbreak{} \mathit{capacity}(X)$
  and set $\mathit{size}(Y) = \mathit{size}(X)$. 
  Copy two (or one, if only one exists) objects from the end of $X$ to the corresponding positions in $Y$,
  decrease $\mathit{size}(X)$ as such and decrease $\mathit{size}(Y)$ by one. 
  If $\mathit{size}(X)$ is zero, release the array $X$ and rename $Y$ as $X$.
\end{description}

Because of the speed objects are copied, a copying process can always
be finished before it will be necessary to start a new copying
process.  Clearly, in connection with each operation, the amount of
work done is $O(1)$. Additionally, since only a constant fraction of
the allocated memory is unused, the amount of space used is
proportional to the number of objects stored in the array.

\subsection{The Priority-Queue Operation Repertoire}

We assume that the atomic components of the priority queues manipulated are
nodes, each storing an element. Further, we assume that the memory
management related to the nodes is done outside the priority queue. For
\Insert{}, the user must give a pointer to a node as its argument; the
node is then made part of the data structure. The ownership of the
node is returned back to the user when it is removed from the data
structure. After such removal, it is the user's responsibility to take
care of the actual destruction of the node.

A meldable (minimum) priority queue should support the following operations. All
parameters and return values are pointers or references to
objects. When describing the effect of an operation, we write (for the
sake of simplicity) ``object $x$'' instead of ``the object pointed to
by $x$''.
\begin{description}
\item[\Construct{}().] Create and return a new priority queue that
  contains no elements.

\item[\Destroy{}($Q$).] Destroy priority queue $Q$ under the
  precondition that the priority queue contains no elements.

\item[\Findmin{}($Q$).] Return the node of a minimum
  element in priority queue $Q$.

\item[\Insert{}($Q$, $x$).] Insert node $x$ (storing an
  element) into priority queue $Q$.

\item[\Decrease{}($Q$, $x$, $v$).] Replace the element
  stored at node $x$ in priority queue $Q$ with element $v$ such
  that $v$ is not greater than the old element.

\item[\Delete{}($Q$, $x$).]  Remove node $x$ from priority queue $Q$.

\item[\Deletemin{}($Q$).] Remove and return the node of a minimum
  element from priority queue $Q$. This operation has the same effect as
  \Findmin{} followed by \Delete{} using its return value as argument.

\item[$\Meld{}(Q_1, Q_2)$.] Move the nodes from priority queues $Q_1$
  and $Q_2$ into a new priority queue and return that priority queue.
  Priority queues $Q_1$ and $Q_2$ are dismissed by this operation.
\end{description}
Observe that it is essential for the operations \Decrease{} and
\Delete{} to take (a handle to) a priority queue as one of their arguments. Kaplan
et al.~\cite{KST02} showed that, if this is not the case, the stated optimal
time bounds are not achievable.

\subsection{Operations in Pseudo-Code}

To represent a priority queue, the following variables are maintained:
\begin{itemize}
\item a pointer to $t_1$ (the root of $T_1$)
\item a pointer to $t_2$ (the root of $T_2$)
\item a pointer to the beginning of the list of active violations of
  $t_1$
\item a pointer to the beginning of the list of inactive violations of
  $t_1$
\item a pointer to the violation structure recording the active violations of $t_1$
\item a pointer to the regular counter keeping track of the children
  of $t_1$; both the rank sequence $\sequence{d_0, d_1,\ldots,
    d_{\mathit{rank}(t_1)}}$ and pointers to the children are
  maintained
\item a pointer to the regular counter keeping track of the children of $t_2$
\end{itemize}

\noindent
Next, we describe the operations in pseudo-code. For this, we use
itemized text, and consciously avoid programming-language details.

\subsection*{\Construct{}()}

\begin{itemize}
\item Initialize all pointers to point to $\mathit{null}$.
\item Create an empty violation array for $t_1$.
\item Create an empty regular counter for $t_1$ and $t_2$.
\end{itemize}

\subsection*{\Destroy{}($Q$)}

\begin{itemize}
\item Destroy the regular counters of $t_1$ and $t_2$.
\item Destroy the violation array for $t_1$.
\item Raise an exception if either $T_1$ or $T_2$ is not empty.
\end{itemize}

\subsection*{\Findmin{}($Q$)}
\begin{itemize}
\item Return $t_1$.
\end{itemize}

\subsection*{\Insert{}($Q$, $x$)}

\begin{itemize}
\item If the element at node $x$ is smaller than that at 
  $t_1$,
 swap $x$ and $t_1$ (but not their violation lists).
\item Make $x$ the child of $t_1$.
\item Update the regular counter of $t_1$ (by performing an {\it increment} at position $0$).
\end{itemize}

\subsection*{\Meld{}($Q_1, Q_2$)}

\begin{itemize}
\item Let the involved roots be $t_1$ and $t_2$ (for $Q_1$), and $t'_1$
  and $t'_2$ (for $Q_2$). Assume without loss of generality that
  the element at $t_1$ is smaller than that at $t'_1$.
\item Make $t_1$ the root of the new first tree of the resulting priority queue.
\item Dismiss the violation array of $t'_1$.
\item Let $s$ denote the node of the highest rank among the roots
  $t_1$, $t_2$, $t'_1$, and $t'_2$.
\item If $s = t_1$
 \begin{itemize}
  \item Move the other roots as the children of $t_1$. 
  \item Update the regular counter of $t_1$ accordingly.
 \end{itemize}
\item Otherwise:
 \begin{itemize}
  \item Make $s$ the root of the new second tree of the resulting priority queue.
  \item Move the other (at most two) roots below $s$, and make these roots violating 
  (by adding them to the violation structure of $t_1$).
  \item Update the regular counter of $s$ accordingly.
  \item Perform two violation reductions, if possible.
 \end{itemize}
\item Dismiss the regular counters of the roots moved below $t_1$ and $s$.
\item Extend the violation array of $t_1$ by $O(1)$ locations, if necessary.
\item If $s \neq t_1$, repeat $O(1)$ times: 
   \begin{itemize}
    \item Move a child of $\mathit{rank}(t_1)$ from below $s$ to below $t_1$.
    \item  Update the regular counter of $t_1$ (by performing an {\it increment} at position $\mathit{rank}(t_1)$).
    \item  Update the regular counter of $s$ (by performing a {\it decrement} at position $\mathit{rank}(t_1)$).
   \end{itemize}
\item If $s \neq t_1$ and $\mathit{rank}(s) \leq \mathit{rank}(t_1)$
	\begin{itemize} 
		\item  Move the whole tree of $s$ below $t_1$.
	  \item  Update the regular counter of $t_1$ (by performing an {\it increment} at position $rank(s)$).
	  \item  Dismiss the regular counter of $s$.
  \end{itemize}
\end{itemize}

\subsection*{\Decrease{}($Q$, $x$, $v$)}

\begin{itemize}
\item Replace the element at node $x$ with element $v$.
\item If the element at $x$ is smaller than that at $t_1$,
 swap $x$ and $t_1$ (but not their violation lists).
A swap should be followed by an update of the external
pointers referring to $x$. If a swap was executed and
 if any of the arrays (regular counters or violation array) had a pointer
  to $x$, update this pointer at $\mathit{rank}(x)$ to point to $t_1$
  instead. 
\item If $x$ is $t_1$, $t_2$, or a child of $t_1$,
 stop.
\item If $x$ was violating,
 remove it from the violation structure where it was in.
\item Make $x$ violating by adding it to the violation structure of
  $t_1$ (either as an active or an inactive violation depending on
  its rank).
\item Perform one violation reduction, if possible.
\item Extend the violation array of $t_1$ by $O(1)$ locations, if necessary.
\item If $T_2$ exists, repeat $O(1)$ times: 
   \begin{itemize}
    \item Move a child of $\mathit{rank}(t_1)$ from below $t_2$ to below $t_1$.
    \item  Update the regular counters $t_1$ (by performing an {\it increment} at position $\mathit{rank}(t_1)$).
    \item  Update the regular counter of $t_2$ (by performing a {\it decrement} at position $\mathit{rank}(t_1)$).
   \end{itemize}
\item If $T_2$ exists and $\mathit{rank}(t_2) \leq \mathit{rank}(t_1)$
	\begin{itemize} 
		\item  Move the whole tree $T_2$ below $t_1$.
	  \item  Update the regular counter of $t_1$ (by performing an {\it increment} at position $rank(t_2)$).
	  \item  Dismiss the regular counter of $t_2$.
  \end{itemize}  
\end{itemize}

\subsection*{\Deletemin{}($Q$)}

\begin{itemize}
\item Merge $t_2$ (as a single node) and all its subtrees with the children of $t_1$,
while extending and updating the regular counter of $t_1$ accordingly.
\item Dismiss the regular counter of $t_2$.
\item Determine the new minimum by scanning the children of
$t_1$ and all violations in the violation lists of $t_1$. Let $x$ be
  the node containing the new minimum.
\item If $x$ is a violation node
 \begin{itemize}
   \item Remove a child of $rank(x)$ from below $t_1$.
   \item Update the regular counter of $t_1$ accordingly.
   \item Make the detached node violating, and attach it in place of $x$.
 \end{itemize}  
\item Otherwise
 \begin{itemize}
   \item Remove $x$ from below $t_1$.
   \item Update the regular counter of $t_1$ accordingly.
 \end{itemize}
\item Merge the children of $x$ with those of $t_1$,
			and update the regular counter accordingly.
\item Append the violation list of $t_1$ to that of $x$.
\item Make the violation array large enough, and add all
  violations of $x$ that are not already there to the violation array.
\item Release $t_1$ and move $x$ to its place.
\item Perform as many violation reductions as possible.
\end{itemize}

\subsection*{\Delete{}($Q$, $x$)}
\begin{itemize}
\item Swap node $x$ and $t_1$.
\item If any of the arrays (regular counters or violation array) had a pointer
  to $x$, update this pointer at $\mathit{rank}(x)$ to point to $t_1$
  instead. 
\item Make the current node $x$ violating.
\item Remove the current root of $T_1$ as in \Deletemin{}. 
\end{itemize}

\end{document}